\documentclass{article}
\usepackage{amsmath,amssymb,amsthm,verbatim}
\usepackage[margin=1in]{geometry}

\title{A Pseudorandom Generator for Polynomial Threshold Functions of Gaussian with Subpolynomial Seed Length}
\author{Daniel M. Kane}

\newcommand{\R}{\mathbb{R}}

\newcommand{\jac}{\textrm{Jac}}

\newcommand{\pr}{\textrm{Pr}}
\newcommand{\sgn}{\textrm{sgn}}
\newcommand{\E}{\mathbb{E}}
\newcommand{\Z}{\mathbb{Z}}

\newtheorem{thm}{Theorem}
\newtheorem{prop}[thm]{Proposition}
\newtheorem{cor}[thm]{Corollary}
\newtheorem{lem}[thm]{Lemma}

\newtheorem*{defn}{Definition}

\begin{document}
\maketitle

\begin{abstract}

\end{abstract}

\section{Introduction}

We say that a function $f:\R^n\rightarrow\{+1,-1\}$ is a degree-$d$ \emph{polynomial threshold function} (PTF) if it is of the form $f(x)=\sgn(p(x))$ for $p$ some (degree-$d$) polynomial in $n$ variables.  Polynomial threshold functions make up a natural class of Boolean functions and have applications to a number of fields of computer science such as circuit complexity \cite{curciutApp}, communication complexity \cite{commApp} and learning theory \cite{learningApp}.

In this paper we study the question of pseudorandom generators for polynomial threshold functions of Gaussians.  In particular, we wish to find explicit functions $F:\{0,1\}^s\rightarrow \R^n$ so that for any degree-$d$ polynomial threshold function $f$
$$
\left| \E_{x\sim_u \{0,1\}^s}[ f(F(x))] - \E_{X\sim \mathcal{G}^n}[f(X)]\right| < \epsilon.
$$
We say that such an $F$ is a pseudorandom generator of seed length $s$ that fools degree-$d$ polynomial threshold functions with respect to the Gaussian distribution to within $\epsilon$.  In this paper, we develop a new such generator whose seed length is $O(\epsilon^{-o(1)})$ for any fixed $d,n$.

\subsection{Previous Work}

There have been a number of previous papers dealing with the question of finding pseudorandom generators for polynomial threshold functions with respect the the Gaussian distribution or the Bernoulli distribution (i.e. uniform over $\{-1,1\}^n$).  Several early works in this area showed that polynomial threshold functions of various degrees could be fooled by arbitrary $k$-wise independent families of Gaussian or Bernoulli random variables.  It should be noted that a $k$-wise independent family of Bernoulli random variables can be generated from a seed of length $O(k\log(n))$.  Although, any $k$-wise independent family of Gaussians will necessarily have infinite entropy, it is not hard to show that a simple discretization of these random variables leads to a generator of comparable seed length.  These results on fooling polynomial threshold functions with $k$-independence are summarized in Table \ref{kIndepTable} below.
\begin{table}[h]\label{kIndepTable}
\begin{tabular}{|l|c|c|l|}
\hline
Paper & Bernoulli/Gaussian & d & k \\
\hline
Diakonikolas, Gopalan, Jaiswal, Servedio, Viola \cite{IndepHalf} & Bernoulli & 1 & $O(\epsilon^{-2}\log^2(\epsilon^{-1}))$ \\
Diakonikolas, Kane, Nelson \cite{IndepDeg2} & Gaussian & 1 & $O(\epsilon^{-2})$ \\
Diakonikolas, Kane, Nelson \cite{IndepDeg2} & Both & 2 & $O(\epsilon^{-8})$\footnotemark\\
Kane \cite{kIndep} & Both & $d$ & $O_d\left( \epsilon^{-2^{O(d)}}\right)$ \\
\hline
\end{tabular}
\end{table}
\footnotetext{The bound in \cite{IndepDeg2} for the Bernoulli case is actually $\tilde O(\epsilon^{-9})$, but this can be easily improved to $O(\epsilon^{-8})$ using technology from \cite{DD}.}
Unfortunately, it is not hard to exhibit $k$-wise independent families of Bernoulli or Gaussian random variables that fail to $\epsilon$-fool the class of degree-$d$ polynomial threshold functions for $k=\Omega(d^2 \epsilon^{-2})$, putting a limit on what can be obtained through mere $k$-independence.

There have also been a number of attempts to produce pseudorandom generators by using more structure than limited independence.  In \cite{MZ}, Meka and Zuckerman develop a couple of such generators in the Bernoulli case.  Firstly, they make use of pseudorandom generators against space bounded computation to produce a generator of seed length $O(\log(n) +\log^2(\epsilon^{-1}))$ in the special case where $d=1$.  By piecing together several $k$-wise independent families, they produce a generator for arbitrary degree PTFs of seed length $2^{O(d)}\log(n) \epsilon^{-8d-3}$.  In \cite{DD}, the author develops an improved analysis of this generator allowing for a seed length as small as $O_{c,d}(\log(n)\epsilon^{-11-c})$.

For the Gaussian case, the author developed a generator of seed length $2^{O_c(d)}\log(n) \epsilon^{-4-c}$ in \cite{GPRG}.  This generator was given essentially as an average several random variables each picked independently from a $k$-wise independent family of Gaussians.  The analysis of this generator was also improved in \cite{DD}, obtaining a seed length of $O_{c,d}(\log(n) \epsilon^{-2-c})$.  In this paper, we improve on this bound further.  We make use of a slight modification of the above generator, by using unequal weights in our averaging process and obtain a seed length of $O_{c,d}(\log(n) \epsilon^{-c})$.

\subsection{Outline of Paper}

In Section \ref{BackgroundSec}, we will introduce some conventions that we will use throughout the paper, and review some basic results on polynomials of Gaussians.

The key idea in our analysis is that for $p$ an approximately linear polynomial that $\E[p(X)]$ is a smooth function in the coefficients of $p$, and thus can be well approximated by a polynomial in these coefficients.  A precise statement of this idea is presented in Proposition \ref{PolyApproxProp}, whose proof takes up most of Section \ref{PolySec}.

 Hence, by the above claim, if $p$ is an approximately linear polynomial, then for $f=\sgn\circ p$, and $Y$ a random variable whose low-degree moments are correct, we will have that $\E[f(\epsilon Y + \sqrt{1-\epsilon^2} X)]$ for $X$ a random Gaussian will be approximately correct.  This is because $p$ can be thought of as a nearly linear polynomial in $X$ whose coefficients are given by polynomials in $Y$.  Proposition \ref{PolyApproxProp} will therefore imply that this expectation is approximated by the expectation of some polynomial in $Y$.

Unfortunately, a generic polynomially will not necessarily be approximately linear.  We fix this by evaluating the polynomial near a random input.  In particular, if we consider $p(\epsilon X_1 + \sqrt{1-\epsilon^2}X_2)$ for a fixed random Gaussian $X_2$, the resulting polynomial in $X_1$ is likely to be approximately linear.  Such an analysis will work for a sufficiently non-singular polynomial (i.e. a polynomial whose derivative is unlikely to be small).  Not all polynomials are non-singular, but as we will show in Section \ref{NonSingSec}, any polynomial can be written in terms of non-singular polynomials.

In Section \ref{FinalSec}, we use this theory to develop a sequence of iteratively more detailed generators eventually leading to one that satisfies our requirements.  Using the ideas above, we show in Proposition \ref{oneStepProp} that  for $X$ a true $n$-dimensional Gaussian and $Y$ a $k$-wise independent family of Gaussians that $\epsilon Y + \sqrt{1-\epsilon^2}X$ produces a PRG that fools degree-$d$ PTFs to within $O_{d,k}(\epsilon^k)$.  Iteratively replacing the $X$ involved by such a generator, we obtain a PRG (see Proposition \ref{epsXProp}) given by
$$
\sum_{i=0}^{\ell-1} \epsilon (1-\epsilon^{2})^{i/2} Y_i + (1-\epsilon^{2})^{\ell/2} X.
$$
It is easy to see that for $\ell$ large, that the $X$ term may safely be removed introducing at most a small error (see Proposition \ref{finalPRGProp}).  Finally, in Theorem \ref{mainThm}, we put these results together to produce a PRG of seed length $O_{c,d}(\log(n)\epsilon^{-c})$.

\section{Background}\label{BackgroundSec}

\subsection{Notation}

We will use the notation $O_a(N)$ to denote a quantity whose absolute value is bounded above by $N$ times some constant depending only on $a$.  Throughout this paper, the variables $X,X_1,\ldots$ will be used to denote multidimensional Gaussian random variables unless stated otherwise.

We recall here the definition of a polynomial threshold function:
\begin{defn}
A function $f:\R^n\rightarrow\{\pm 1 \}$ is a (degree-$d$) \emph{polynomial threshold function} (or PTF) if it is of the form $f(x) = \sgn(p(x))$ for some (degree-$d$) polynomial $p$.
\end{defn}

Another important definition will be the following:

\begin{defn}
We say that a random variable $Y$ taking values in $\R^n$ is a \emph{$k$-design}, if all of the moments of $Y$ of order at most $k$ agree with the corresponding moments of a standard $n$-dimensional Gaussian.
\end{defn}
Note that any $k$-wise independent family of Gaussians is a $k$-design.  Also note that applying any orthogonal transformation to a $k$-design yields another $k$-design.  Throughout this paper we will use the variables $Y,Y_1,Y_i,\ldots$ to denote $k$-designs for some $k$ unless otherwise specified.

\subsection{Polynomials of Gaussians}

We recall some basic facts about polynomials of Gaussians.  We begin by recalling the $L^t$-norm of a function.

\begin{defn}
For a function $p:\R^n\rightarrow\R$, we let
$$
|p|_t = \left( \E_X[|p(X)|^t ] \right)^{1/t}.
$$
\end{defn}

We now recall some basic distributional results about polynomials evaluated at random Gaussians.

\begin{lem}[Carbery and Wright]\label{anticoncentrationLem}
If $p$ is a degree-$d$ polynomial then
$$
\pr(|p(X)| \leq \epsilon|p|_2) = O(d\epsilon^{1/d}).
$$
Where the probability is over $X$, a standard $n$-dimensional Gaussian.
\end{lem}

We will make use of the hypercontractive inequality.  The proof follows from Theorem 2 of \cite{hypercontractivity}.

\begin{lem}\label{hypercontractiveLem}
If $p$ is a degree-$d$ polynomial and $t>2$, then
$$
|p|_t \leq \sqrt{t-1}^d |p|_2.
$$
\end{lem}

In particular this implies the following concentration bound:

\begin{cor}\label{ConcentrationCor}
If $p$ is a degree-$d$ polynomial and $N>0$, then
$$
\pr_X(|p(X)| > N|p|_2) = O\left(2^{-(N/2)^{2/d}} \right).
$$
\end{cor}
\begin{proof}
Apply the Markov inequality and Lemma \ref{hypercontractiveLem} with $t = (N/2)^{2/d}$.
\end{proof}

\subsection{Orthogonal Polynomials}

We recall that the orthogonal polynomials form an orthonormal basis of the set of polynomials with respect to the Gaussian inner product.  Thus any polynomial can be written uniquely as a linear combination of orthogonal polynomials
$$
p(x) = \sum_{a\in \Z_{\geq 0}^n} c_a(p) h_a(x).
$$
We let
$$
p^{[k]}(x) := \sum_{|a|_1 = k} c_a(p) h_a(x)
$$
be the sum of the terms in the above decomposition consisting of orthogonal polynomials of degree exactly $k$.  Furthermore, we let
$$
p^{[\geq k]} := \sum_{m\geq k} p^{[m]}.
$$

We recall from \cite{DD} that
$$
\E\left[\left|\partial_{X_1}\cdots\partial_{X_\ell} p(X) \right|^2 \right] = \sum_{k} k(k-1)\cdots(k-\ell+1) \left| p^{[k]} \right|_2^2.
$$
Where $\partial_{X_i}$ above denotes the directional derivative in the $X_i$ direction for $X_i$ a random Gaussian.

\section{Polynomial Approximation of Expectations}\label{PolySec}

In this Section, we prove the following Proposition, which says that the expectation of a threshold function of a polynomial $p$, that is approximately linear can be approximated by a polynomial in the coefficients of $p$.

\begin{prop}\label{PolyApproxProp}
Let $d,m$ and $k$ be positive integers.  Let $p:\R^m\rightarrow \R^m$ be a degree-$d$ polynomial given by $p(x) = x + q(x)$.  Let $\epsilon,N>0$ be real numbers and let $f:\R^m \rightarrow [-1,1]$ be any function.  Then there exists a polynomial $R$ in the coefficients of $q$ of degree less than $k$, dependent only on $d,m,k,\epsilon$ and $f$ so that for all $q$
$$
\E[f(p(X))] = R(q) + O_{d,m,k,N}(\epsilon^N + \epsilon^{-1} |q|^k).
$$
Where above $|q|$ denotes the largest absolute value of a coefficient of $q$.  Furthermore, using the same notation, $|R| \leq \log(\epsilon^{-1})^{O_{d,m,k}(1)}.$
\end{prop}

In order to expand upon the intuition behind Proposition \ref{PolyApproxProp}, we begin by sketching the proof in the case that $m=d=1$.  In this case we may write $q(x)=ax+b$.  It is then the case that
$$
\E[f(p(X))] = \E[f((1+a)X+b)] = \frac{1}{\sqrt{2\pi}}\int_{-\infty}^\infty f((1+a)x+b)e^{-x^2/2}dx.
$$
The key idea is to evaluate the above by making the change of variables $y=(1+a)x+b.$  The above is then equal to
$$
\frac{1}{\sqrt{2\pi}}\int_{-\infty}^\infty f(y) e^{-((y-b)/(1+a))^2/2} (1+a)^{-1} dy.
$$
For small $a$ and $b$, we may approximate the integrand above by a degree $k-1$ Taylor polynomial in $a$ and $b$ introducing an error on the order of $|q|^k$ in the process.  Integrating then yields a polynomial in $a$ and $b$ plus a small error.  The proof of Proposition \ref{PolyApproxProp} is a straightforward generalization of this idea, though we will see some technical difficulties arising from the more complicated change of variables, and the necessity of keeping better track of errors.

\begin{proof}
Note that $|\E[f(p(X))]|\leq 1$, therefore we may assume that $\epsilon \ll 1$, or there is nothing to prove.
Similarly, we may assume that $|q|\ll \epsilon^{1/(2k)}$ or else $\epsilon^{-1}|q|^k \gg |R(q)|+1$ and there is again nothing to prove.  In particular, we may assume that for $c(d,m)$ a sufficiently small constant (in terms of $d$ and $m$) that $\epsilon < c(d,m),$ and $|q| < c(d,m)\epsilon^{1/(2k)}.$

Note that
$$
\E[f(p(X))] = \int_{\R^m} f(p(x))\phi(x)dx
$$
where $\phi(x)=(2\pi)^{-m/2}e^{-\frac{|x|_2^2}{2}}$.  Up to an error of $O_{m,N}(\epsilon^N)$, we may ignore the integral outside of the range where $|x|_2 \leq \log(\epsilon^{-1})$.  Note furthermore, that in this range, for $\epsilon$ and $|q|$ sufficiently small, we have $$|p(x)| \leq |x| + |q(x)| \leq \log(\epsilon^{-1}) + O_{d,m}(|q|\log(\epsilon^{-1})^d) \leq 2\log(\epsilon^{-1}).$$

We claim that in this range of inputs and outputs that $p$ has a nice inverse.  In particular, if $y\in \R^m$ with $|y|_2 \leq 2\log(\epsilon^{-1})$, we claim that there is a unique $x$ with $|x|_2\leq 3 \log(\epsilon^{-1})$ so that $p(x)=y$.  To show this, we consider the map $M$ from the ball of radius $3\log(\epsilon^{-1})$ to $\R^m$ given by
$$
M(x) = y - q(x).
$$
Again, if $\epsilon$ and $|q|$ are sufficiently small, then
$$
|M(x)| \leq |y| + |q(x)| \leq 2\log(\epsilon^{-1}) + O_{d,m}(|q|\log(\epsilon^{-1})^d) \leq 3\log(\epsilon^{-1}),
$$
and thus $M$ maps the ball of radius $3\log(\epsilon^{-1})$ to itself. For $|x|\leq 3\log(\epsilon^{-1})$, we have that $|q'(x)|$ is bounded by $O_{d,m}(|q||x|^{d-1})$.  For $|q|$ a sufficiently small multiple of $\epsilon^{1/(2k)}$, this is strictly less than $1/2$.  Thus $M$ is a contraction mapping and thus has a unique fixed point.  On the other hand, $M(x)=x$ if and only if $p(x)=y$.  Therefore, for such $y$, we have a unique inverse.  We may now write our expectation as
$$
\E[f(p(X))] = \int_{\substack{|x|_2 \leq 3\log(\epsilon^{-1})\\ |p(x)|_2 \leq 2\log(\epsilon^{-1})}} f(p(x))\phi(x)dx + O_{m,N}(\epsilon^N).
$$

Our plan is now to compute this integral by making the change of variables $y=p(x)$.  We know from the above that in the domain of interest there is a function $p^{-1}$, which by the Inverse Function Theorem is necessarily smooth.  Thus,
$$
\E[f(p(X))] = \int_{|y|_2 \leq 2\log(\epsilon^{-1})} f(y) \left(\frac{\phi(p^{-1}(y))}{|\jac(p(x))|_{x=p^{-1}(y)}}\right)dy + O_{m,N}(\epsilon^N).
$$

The fundamental idea of our proof will be to approximate $\left(\frac{\phi(p^{-1}(y))}{|\jac(p(x))|_{x=p^{-1}(y)}}\right)$ by a polynomial in $q$ with coefficients depending on $y$.  Integrating the above formula for $\E[f(p(X))]$, will then yield our result.  We recall that $M$ was a contraction mapping with constant $O_{d,m}(|q|\log(\epsilon^{-1})^{d-1})<1/2$, and that $p^{-1}(y)$ is the fixed point of $M$.  Since $|M(y)-y| = O_{d,m}(|q|\log(\epsilon^{-1})^d)$, we have that $|p^{-1}(y)-y| = O_{d,m}(|q|\log(\epsilon^{-1})^d)$.  Let $M^{\ell}$ be the $\ell^{\textrm{th}}$ iterate of $M$.  Since $M$ is a contraction mapping with constant $O_{d,m}(|q|\log(\epsilon^{-1})^{d-1})$ and fixed point $p^{-1}(y)$, we have that $$|M^{\ell}(y)-p^{-1}(y)| = O_{d,m,\ell}(|q|^{\ell+1}\log(\epsilon^{-1})^{d+d\ell}).$$  Notice that for fixed $y$ that $M^\ell(y)$ is a polynomial in $q$ of degree at most $d^\ell$, whose coefficients have size at most $O_{d,m,\ell}((1+|y|_2)^{d^\ell})$.

We may Taylor expand $\phi(x)$ about $x=y$ to obtain an expression $$\phi(x) = T_{k,y}((x-y)) + O_{m,k}(|x-y|_2^k),$$ where $T_{k,y}$ is a polynomial of degree less than $k$ with coefficients of size $O_{m,k}(1)$.  Similarly, we may write $|\jac(p(x))|$ as a polynomial in $x$ and $q$ that is equal to $1+O_{d,m}(|q|(1+|x|)^{d(d-1)}).$  We may therefore Taylor expand its inverse as $$\frac{1}{|\jac(p(x))|} = S_k(x,q) + O_{d,m,k}(|q|^k(1+|x|)^{kd(d-1)}), $$ where $S_k$ is a polynomial of degree at most $k(d+1)$ and coefficients of size $O_{d,m,k}(1)$.

Putting the above together, we have that:
\begin{align*}
& \left(\frac{\phi(p^{-1}(y))}{|\jac(p(x))|_{x=p^{-1}(y)}}\right)\\ = & \left(T_{k,y}(p^{-1}(y)-y)+O_{m,k}(|q|^k\log(\epsilon^{-1})^{kd})\right)\left(S_k(p^{-1}(y),q) + O_{d,m,k}(|q|^k\log(\epsilon^{-1})^{kd(d-1)})\right)\\
= & T_{k,y}(M^k(y)-y)S_k(M^k(y),q) + O_{d,m,k}\left(|q|^k \log(\epsilon^{-1})^{O_{d,m,k}(1)} \right)\\
= & R_y(q) + O_{d,m,k}\left(|q|^k \log(\epsilon^{-1})^{O_{d,m,k}(1)} \right).
\end{align*}
Where above $R_y(q)$ is some polynomial in $q$ of degree $O_{d,m,k}(1)$ with coefficients dependent on $y$ and of size at most $\log(\epsilon^{-1})^{O_{d,m,k}(1)}$.  By absorbing the terms of $R_y$ of degree at least $k$ into the error, we may assume that $R$ has degree strictly less than $k$.  Therefore, we have that
\begin{equation}\label{unintegratedPolyApproxEqn}
\E[f(p(X))] = \int_{|y|\leq 2\log(\epsilon^{-1})} f(y) (R_y(q) + O_{d,m,k}(\log(\epsilon^{-1})^{O_{d,m,k}(1)}|q|^k)) dy + O_{m,N}(\epsilon^N).
\end{equation}
Letting,
$$
R(q) := \int_{|y|\leq 2\log(\epsilon^{-1})} f(y) R_y(q) dy,
$$
we have by Equation (\ref{unintegratedPolyApproxEqn}) that (noting that the domain of integration has volume at most $(4\log(\epsilon^{-1}))^m$)
$$
\E[f(p(X))] = R(q) + O_{d,m,k,N}(\epsilon^{-1}|q|^k + \epsilon^N).
$$
Thus completing our proof.

\end{proof}

We can use Proposition \ref{PolyApproxProp} to analyze a simple form of our generator.

\begin{prop}\label{LinearizedGeneratorProp}
Let $p$ be a degree-$d$ polynomial that can be written in the form $p(x)=h(q_1(x),\ldots,q_m(x))$ for some function $h$ and some polynomials $q_i$ of degree at most $d$.  Let $f(x)=\sgn(p(x))$ be the corresponding polynomial threshold function.  Suppose that for each $i$ that $q_i(x) = x_i + r_i(x)$ for some polynomial $r_i$.  Let $\epsilon>0$ be a real number and $k$ be an even integer.  Let $X$ be a random Gaussian and $Y$ a $kd$-design that is independent of $X$.  Then
$$
\left|\E[f(X)] - \E\left[f\left(\epsilon Y + \sqrt{1-\epsilon^2}X \right) \right] \right| = O_{d,m,k}\left(\epsilon^{k-1} + \epsilon^{-1}\sum_{i=1}^m |r_i|_2^k\right).
$$
\end{prop}
\begin{proof}
Note that $X$ can be written as the sum $\epsilon X_1 + \sqrt{1-\epsilon^2}X_2$ for $X_1$ and $X_2$ independent Gaussians.  Hence it suffices to show that $\E\left[f\left(\epsilon Y + \sqrt{1-\epsilon^2}X \right) \right]$ is determined to within $O_{d,m,k}\left(\epsilon^{k-1} + \epsilon^{-1}\sum_{i=1}^m |r_i|_2^k\right)$ simply be the low degree moments of $Y$.

We may rewrite $X$ as $(X_0,X_1)$, where $X_0$ is the Gaussian given by the first $m$ coordinates of $X$ and $X_1$ consists of the remaining coordinates.  We let $Q(x_0,x_1,y)$ be the vector-valued polynomial given by
\begin{align*}
Q(X_0,X_1,Y)_i & = \frac{q_i\left(\epsilon Y + \sqrt{1-\epsilon^2}(X_0,X_1) \right)}{\sqrt{1-\epsilon^2}}\\
& = (X_0)_i + \left(\frac{\epsilon Y_i + r_i\left(\epsilon Y + \sqrt{1-\epsilon^2}(X_0,X_1) \right)}{\sqrt{1-\epsilon^2}} \right).
\end{align*}
Upon fixing values for $Y$ and $X_1$ we let $q^{Y,X_1}(X_0)$ be the vector valued polynomial given by
$$
q^{X_1,Y}_i(X_0) := \left(\frac{\epsilon Y_i + r_i\left(\epsilon Y + \sqrt{1-\epsilon^2}(X_0,X_1) \right)}{\sqrt{1-\epsilon^2}} \right).
$$
We have that
$$
Q(X_0,X_1,Y) = X_0 + q^{X_1,Y}(X_0).
$$

We have that
\begin{align*}
\E\left[ f\left(\epsilon Y + \sqrt{1-\epsilon^2}X \right)\right] & = \E\left[\sgn\left(h\left(\sqrt{1-\epsilon^2}Q(X_0,X_1,Y) \right) \right) \right]\\
& = \E\left[g(Q(X_0,X_1,Y)) \right]\\
& = \E_{X_1,Y}[\E_{X_0}[g(Q(X_0,X_1,Y))]]\\
& = \E_{X_1,Y}[R(q^{X_1,Y}) + O_{d,m,k}(\epsilon^{-1}|q^{X_1,Y}|^k + \epsilon^k)].
\end{align*}
Where $g$ above is given by $g(x) = \sgn(h(\sqrt{1-\epsilon^2}x))$, and $R$ is the appropriate polynomial given by Proposition \ref{PolyApproxProp}.  Since the expectation of $R(q^{X_1,Y})$ is determined the moments $Y$ up to degree $kd$, this expectation is determined up to an error of
$$
O_{d,m,k}\left(\epsilon^k + \epsilon^{-1}\E[|q^{X_1,Y}|^k] \right).
$$
We note that $|q^{X_1,Y}| = O_{d,m}(|q^{X_1,Y}|_2) = O_{d,m,k}(|q^{X_1,Y}|_k).$  Therefore the error above is
\begin{align*}
& O_{d,m,k}\left(\epsilon^k + \epsilon^{-1}\E[|q^{X_1,Y}(X_0)|^k]  \right)\\
= & O_{d,m,k}\left(\epsilon^k + \epsilon^{-1}\sum_{i=1}^m\E[|\epsilon Y_i + r_i\left(\epsilon Y + \sqrt{1-\epsilon^2}(X_0,X_1) \right)|^k]  \right)\\
= & O_{d,m,k}\left(\epsilon^k + \epsilon^{-1}\left(\epsilon^k + \sum_{i=1}^m |r_i|_2^k \right)  \right)\\
= & O_{d,m,k}\left(\epsilon^{k-1} + \epsilon^{-1}\sum_{i=1}^m |r_i|_2^k\right).
\end{align*}
Where the second to last line above is by Lemma \ref{hypercontractiveLem} and the fact that $Y$ is a $kd$-design.
\end{proof}

\section{Non-Singular Sets}\label{NonSingSec}

Our basic plan will be to use Proposition \ref{LinearizedGeneratorProp} to show that the generator $\epsilon Y + \sqrt{1-\epsilon^2}X$ fools all polynomial threshold functions. The idea will be to let $\sqrt{1-\epsilon^2}X = \sqrt{\epsilon}X_1 + \sqrt{1-\epsilon-\epsilon^2}X_2$ for $X_1$ and $X_2$ independent Gaussians.  Upon fixing a random value for $X_2$, it is not hard to show that the resulting polynomial of $\epsilon Y + \sqrt{\epsilon} X_1$ will likely have its quadratic terms of size $\tilde O(\epsilon)$.  Were it the case that the linear term of this polynomial were $\Theta(\sqrt{\epsilon})$, (as seems likely) we could apply Proposition \ref{LinearizedGeneratorProp} almost immediately.  Unfortunately, if this polynomial has essentially no linear terms, this technique may fail.  The possibility of this failure is closely related to our original polynomial having small derivatives near $X_2$.  We will want to consider polynomials for which this does not happen with non-negligible probability.

\begin{defn}
Given a sequence of polynomials $(q_1,\ldots,q_m)$, we say that they form an \emph{$(\epsilon,c,N)$-non-singular set} if
$$
\pr_X\left(\left|\bigwedge_{j}\partial q_j(X) \right|_2 < \epsilon^c \right) < \epsilon^N.
$$

We recall the definition from \cite{DD} that for a degree-$d$ polynomial $p:\R^n\rightarrow\R$, we say that a set of polynomials $(h,q_1,\ldots,q_m)$ is a \emph{decomposition of $p$ of size $m$} if $q_i:\R^n\rightarrow \R$, and $h:\R^m\rightarrow\R$ are polynomials so that
\begin{itemize}
\item $p(x) = h(q_1(x),\ldots,q_m(x))$
\item For every monomial $\prod x_i^{a_i}$ appearing in $h$, we have that $\sum a_1 \deg(q_i) \leq d$
\end{itemize}

Furthermore, we say that a polynomial $p$ has an \emph{$(\epsilon,c,N)$-non-singular decomposition of size $m$} if $p$ has a decomposition $(h,q_1,\ldots,q_m)$ with $|q_i|_2\leq 1$ for all $i$ and so that $(q_1,\ldots,q_m)$ is an $(\epsilon,c,N)$-non-singular set.
\end{defn}

The key fact about these decompositions that we will need is the following structure theorem.

\begin{thm}\label{nonSingDecompThrm}
Let $p$ be a degree-$d$ polynomial, and let $\epsilon,c,N>0$.  Then there exists a degree-$d$ polynomial $p_0$ with $|p-p_0|_2 = O_{c,d,N}(\epsilon^N)|p|_2$ so that $p_0$ has an $(\epsilon,c,N)$-non-singular decomposition of size $O_{c,d,N}(1)$.
\end{thm}
\begin{proof}
This follows from the proof of the Diffuse Decomposition Theorem of \cite{DD}.
\end{proof}

\section{The PRG}\label{FinalSec}

In this Section, we will prove a sequence of increasingly more powerful results for PRGs.  We begin by showing that if our polynomial has a non-singular decomposition that $\epsilon Y + \sqrt{1-\epsilon^2}X$ is an appropriate generator.

\begin{prop}\label{RegPRGProp}
Let $d,k$ be integers and $\epsilon>0$.  Let $p$ be a degree-$d$ polynomial with an $(\epsilon,1/10,k)$-non-singular decomposition of size $m$.  Let $f$ be the corresponding polynomial threshold function.  Let $X$ be a Gaussian, and $Y$ a $10kd$-design independent of $X$.  Then
$$
\left|\E[f(X)] - \E\left[f\left( \epsilon Y + \sqrt{1-\epsilon^2} X \right) \right]\right| = O_{d,m,k}(\epsilon^k).
$$
\end{prop}
\begin{proof}
First we assume that $\epsilon$ is sufficiently small given $d,m$ and $k$, for otherwise there is nothing to prove.

It suffices to show that the expectation of $f\left( \epsilon Y + \sqrt{1-\epsilon^2} X \right) $ is determined to within $O_{d,m,k}(\epsilon^k)$ by the low order moments of $Y$.

Let $p$ have the $(\epsilon,1/10,k)$-non-singular decomposition $(h,q_1,\ldots,q_m)$.  Write $\sqrt{1-\epsilon^2}X$ as $\sqrt{\epsilon}X_1+\sqrt{1-\epsilon-\epsilon^2}X_2$ for $X_1$ and $X_2$ independent Gaussians. Let $\epsilon X_0 + \sqrt{\epsilon} X_1 = \sqrt{\epsilon + \epsilon^2}Z$ for $X_0$ an independent Gaussian, and $W=\epsilon X_0 +\sqrt{1-\epsilon^2}X$.  Consider each of the $q_i$ as functions of $Z$ and $X_2$.  Thinking of $X_2$ as fixed let $q_i^{X_2}(Z) = q_i(X_2,Z)$.  Notice that
\begin{align*}
\E_{X^2}\left[ \left| \left(q_i^{X^2}\right)^{[\geq 2]} \right|_2^2 \right] & \leq \E_{X_2,Z,X_3,X_4}[|\partial_{X_3}^Z \partial_{X_4}^Zq_i(X_2,Z)|_2^2  ]\\
& = (\epsilon + \epsilon^2)^2 \E[|\partial_{X_3}^W \partial_{X_4}^Wq_i(W)|_2^2  ]\\
& = O(d^2 \epsilon^2 |q_i|_2^2)\\
& = O(d^2\epsilon^2).
\end{align*}
Where $\partial_{X_i}^Z$ above denotes the directional derivative of with respect to $Z$ in the direction of $X_i$.
Thus, since $\left| \left(q_i^{X_2}\right)^{[\geq 2]} \right|_2^2$ is given by a polynomial in $X_2$, we have by Corollary \ref{ConcentrationCor} that with probability $1-O_{d,m,k}(\epsilon^k)$ that $\left| \left(q_i^{X_2}\right)^{[\geq 2]} \right|_2 \leq \epsilon\log(\epsilon^{-1})^d$ for all $i$.  Similarly, we may show that with this same probability that  $\left| \left(q_i^{X_2}\right)^{[1]} \right|_2 \leq \sqrt{\epsilon}\log(\epsilon^{-1})^d$ for all $i$.  For $X_2$ fixed, let $L_i := \left(q_i^{X_2}\right)^{[1]}.$

Note that with high probability
$$
\partial^Z q_i(Z,X_2) = \partial^Z L_i(Z) + O(\epsilon\log(\epsilon^{-1})^d).
$$
on the other hand, we have that
$$
\partial^Z q_i(Z,X_2) = \sqrt{\epsilon+\epsilon^2}\partial^W q_i(W).
$$
By non-singularity this means that with probability $1-O_{d,m,k}(\epsilon^k)$ we have
$$
\left|\bigwedge_{i} \left( \partial L_i(Z)+ O(\epsilon\log(\epsilon^{-1})^d) \right) \right|_2^2 > \epsilon^{m+1/5}.
$$
On the other hand, the left hand side of the above is
$$
\left|\bigwedge_{i} \left( \partial L_i(Z) \right) \right|_2^2 +O_{d,m}(\epsilon^{m+1/2}\log(\epsilon^{-1})^{2dm}).
$$
Thus for $\epsilon$ sufficiently small, we have with probability at least $1-O_{d,m,k}(\epsilon^k)$ over the choice of $X_2$ that
$$
\left|\bigwedge_{i} \left( \partial L_i(Z) \right) \right|_2^2 > \frac{\epsilon^{m+1/5}}{2}.
$$
If this is the case, then the product of the singular values of the matrix with rows given by the gradients of the $L_i$ is at least $\epsilon^{m/2+1/10}$.  Since none of the singular values can be larger than $O_m(\epsilon^{1/2}\log(\epsilon^{-1})^d)$, this implies that all of the singular values of this matrix are at least $\epsilon^{1/4}$.  Thus if the $q_i$ are replaced by a appropriate linear combinations of their old values (with coefficients at most $\epsilon^{-3/4}$) we can ensure that the $\partial L_i(Z)$ are orthonormal.  By making an appropriate change of variables for $Z$, we may assume that $L_i(Z)=Z_i$.  Removing the degree-$0$ harmonic part of $q_i^{X^2}$, we may assume that $q_i^{X_2}(Z)=Z_i + r_i(Z)$ with $|r_i|_2 = O_m(\epsilon^{1/4} \log(\epsilon^{-1})^d)$.

To summarize, with probability at least $1-O_{d,m,k}(\epsilon^k)$ over the choice of $X_2$, there is an orthogonal change of variables for $Z$, and a sequence of polynomials $q_i',r_i$ with $q_i'(Z) = Z_i + r_i(Z)$ and $|r_i(Z)|_2 =O_m(\epsilon^{1/4} \log(\epsilon^{-1})^d)$ so that $p(Z,X_2)$ has a decomposition into the $q_i'$.  Applying Proposition \ref{LinearizedGeneratorProp}, we find that with probability $1-O_{d,m,k}(\epsilon^k)$ over $X_2$ we have that:
\begin{align*}
&\left| \E_Z\left[ f\left( \sqrt{\epsilon + \epsilon^2}Z + \sqrt{1-\epsilon-\epsilon^2}X_2\right)\right] - \E_{Y,X_1}\left[ f\left( \epsilon Y + \sqrt{\epsilon} X_1 + \sqrt{1-\epsilon-\epsilon^2}X_2\right)\right]\right|\\
&  \ \ \ \ \ = O_{d,k,m}(\epsilon^{5k-1}+(\epsilon^{1/4}\log(\epsilon^{-1})^d)^{10k}) = O_{d,k,m}(\epsilon^k).
\end{align*}
Taking an expectation over $X_2$ completes our proof.
\end{proof}

Next we use Theorem \ref{nonSingDecompThrm} to extend Proposition \ref{RegPRGProp} to arbitrary polynomial threshold functions.

\begin{prop}\label{oneStepProp}
Let $f$ be a degree-$d$ polynomial threshold function.  Let $\epsilon>0$ and $k$ be an integer.  Let $X$ be a random Gaussian and $Y$ a $10kd$-design independent of $X$. It is the case that
$$
\left|\E[f(X)] - \E\left[f\left( \epsilon Y + \sqrt{1-\epsilon^2} X \right) \right]\right| = O_{d,k}(\epsilon^k).
$$
\end{prop}
\begin{proof}

Let $f=\sgn(p(x))$ for some degree-$d$ polynomial $p$ with $|p|_2=1$.  By Theorem \ref{nonSingDecompThrm}, there exists a degree-$d$ polynomial $p_0$ so that $|p-p_0|_2 =O_{d,k}(\epsilon^{2kd+k}) $ so that $p_0$ has an $(\epsilon,1/10,k)$-non-singular decomposition of size $m=O_{d,k}(1)$.  Since $\epsilon Y + \sqrt{1-\epsilon^2}X$ is a $2d$-design, we have by the Markov bound that with probability $1-O_{d,k}(\epsilon^k)$ that
$$
\left|p\left(\epsilon Y + \sqrt{1-\epsilon^2}X \right)-p_0\left(\epsilon Y + \sqrt{1-\epsilon^2}X \right) \right|\leq \epsilon^{kd}.
$$
Note that the polynomials $p_0\pm \epsilon^{kd}$ also have $(\epsilon,1/10,k)$-non-singular decompositions of size $m$.  Therefore, we have by the above, Proposition \ref{RegPRGProp} and Lemma \ref{anticoncentrationLem} that
\begin{align*}
\E\left[f\left( \epsilon Y + \sqrt{1-\epsilon^2} X \right) \right] & = \E\left[\sgn\left(p\left( \epsilon Y + \sqrt{1-\epsilon^2} X \right) \right)\right]\\
& \leq \E\left[\sgn\left(p_0\left( \epsilon Y + \sqrt{1-\epsilon^2} X \right)+\epsilon^{kd} \right)\right]+O_{d,k}(\epsilon^k)\\
& = \E[\sgn(p_0(X)+\epsilon^{kd})]+O_{d,k}(\epsilon^k)\\
& = \E[\sgn(p_0(X)-\epsilon^{kd})]+O_{d,k}(\epsilon^k)\\
& \leq \E[\sgn(p(X))] + O_{d,k}(\epsilon^k)\\
& = \E[f(X)] + O_{d,k}(\epsilon^k).
\end{align*}
And the other direction of the inequality follows analogously.
\end{proof}

Iterating applying Proposition \ref{oneStepProp} yields the following:

\begin{prop}\label{epsXProp}
Let $f$ be a degree-$d$ polynomial threshold function and $\epsilon>0$.  Let $k$ and $\ell$ be integers.  For $1\leq i\leq \ell$ let $Y_i$ $10kd$-designs and $X$ a Gaussian so that $X$ and the $Y_i$ are independent.  Then
\begin{align*}
\left| \E[f(X)] - \E\left[ f\left( \sum_{i=1}^\ell \epsilon \left(\sqrt{1-\epsilon^2}\right)^{i-1} Y_i + \left(\sqrt{1-\epsilon^2}\right)^\ell X\right) \right]\right| = O_{d,k}(\ell \epsilon^k).
\end{align*}
\end{prop}
\begin{proof}
The proof is by induction on $\ell$ and noting that by fixing the values of $Y_1,\ldots,Y_{\ell-1}$ Proposition \ref{oneStepProp} implies that
\begin{align*}
& \left| \E\left[ f\left( \sum_{i=1}^{\ell-1} \epsilon \left(\sqrt{1-\epsilon^2}\right)^{i-1} Y_i + \left(\sqrt{1-\epsilon^2}\right)^{\ell-1} X\right) \right] - \E\left[ f\left( \sum_{i=1}^\ell \epsilon \left(\sqrt{1-\epsilon^2}\right)^{i-1} Y_i + \left(\sqrt{1-\epsilon^2}\right)^\ell X\right) \right]\right|\\ & \ \ \ \ \ \ \ \ \ \ = O_{d,k}( \epsilon^k).
\end{align*}
\end{proof}

It is not hard to get rid of the $X$ in the above generator

\begin{prop}\label{finalPRGProp}
Let $f$ be a degree-$d$ polynomial threshold function and $\epsilon>0$.  Let $k$ and $\ell$ be integers.  For $1\leq i\leq \ell$ let $Y_i$ be independent $10kd$-designs and $X$ a Gaussian.  Then
$$
\left| \E[f(X)] - \E\left[ f\left( \frac{\sum_{i=1}^\ell\left(\sqrt{1-\epsilon^2}\right)^{i-1} Y_i}{\sqrt{\sum_{i=1}^\ell (1-\epsilon^2)^i}} \right) \right]\right| = O_{d,k}\left(\ell \epsilon^k + (1-\epsilon^2)^{\frac{\ell}{2d+1}}\right).
$$
\end{prop}
\begin{proof}
Let $f(x)=\sgn(p(x))$ for $p$ a degree-$d$ polynomial with $|p|_2=1$.

Let
$$
Y := \frac{\sum_{i=1}^\ell\left(\sqrt{1-\epsilon^2}\right)^{i-1} Y_i}{\sqrt{\sum_{i=1}^\ell (1-\epsilon^2)^i}}.
$$
Assume that $X$ and $Y$ are independent and let
$$
Z := \sqrt{1-(1-\epsilon^2)^\ell}Y + \left( \sqrt{1-\epsilon^2} \right)^\ell X.
$$
It is not hard to show that since $Y$ is a $2d$-design that
$$
\E[|p(Y)-p(Z)|^2] = O_d((1-\epsilon^2)^\ell).
$$
Thus by the Markov inequality, with probability at least $1-O_d\left((1-\epsilon^2)^{\frac{\ell}{2d+1}}\right)$, we have that
$$
|p(Y)-p(Z)| \leq (1-\epsilon^2)^{\frac{d\ell}{2d+1}}.
$$
Therefore, we have that
\begin{align*}
\E[f(Y)] & = \E[\sgn(p(Y))]\\
& \leq \E\left[ \sgn\left(p(Z) + (1-\epsilon^2)^{\frac{d\ell}{2d+1}}\right)\right] + O_d\left((1-\epsilon^2)^{\frac{\ell}{2d+1}}\right)\\
& \leq \E\left[ \sgn\left(p(X) + (1-\epsilon^2)^{\frac{d\ell}{2d+1}}\right)\right] + O_{d,k}\left(\ell \epsilon^k + (1-\epsilon^2)^{\frac{\ell}{2d+1}}\right)\\
& \leq \E\left[ \sgn(p(X))\right] + O_{d,k}\left(\ell \epsilon^k + (1-\epsilon^2)^{\frac{\ell}{2d+1}}\right)\\
& \leq \E\left[ f(X)\right] + O_{d,k}\left(\ell \epsilon^k + (1-\epsilon^2)^{\frac{\ell}{2d+1}}\right).
\end{align*}
The other direction of the inequality holds analogously.
\end{proof}

We can finally prove our main result:
\begin{thm}\label{mainThm}
For $d,k$ positive integers and $\epsilon>0$, there exists an explicit pseudorandom generator, $Y$ of seed length $O_{d,k}(\log(n)\epsilon^{-1})$ so that for $X$ an $n$-dimensional Gaussian, and $f$ any degree-$d$ polynomial threshold function in $n$ variables, then
$$
\left| \E[f(X)] - \E[f(Y)] \right| = O_{d,k}(\epsilon^k).
$$
\end{thm}
\begin{proof}
Let $\delta=\epsilon^{1/3}$.  Let $\ell = \delta^{-2}\log(\epsilon^{-k(2d+1)})$.  Let $Z_1,\ldots,Z_\ell$ be independent $10d(3k+3)$-designs.  Let
$$
Z := \frac{\sum_{i=1}^\ell\left(\sqrt{1-\delta^2}\right)^{i-1} Z_i}{\sqrt{\sum_{i=1}^\ell (1-\delta^2)^i}}.
$$
By Proposition \ref{finalPRGProp} we have that
$$
\left| \E[f(X)] - \E[f(Z)] \right| = O_{d,k}\left(\ell \delta^{3k+3} + (1-\delta^2)^{\frac{\ell}{2d+1}}\right)=O_{d,k}(\epsilon^k).
$$

By Gauss-Jacobi quadrature, there is a $1$-dimensional $10d(3k+3)$-design supported on a set of size $10d(3k+3)$.  Therefore there is an explicit random variable with seed $O_{d,k}(\log(n/\epsilon))$ which differs from this by at most $\epsilon^k n^{-1} \ell^{-1}$ in statistical distance.  A $10d(3k+3)$-wise-independent family of $n$ of these variables, has seed length $O_{d,k}(\log(n/\epsilon))$ and is within a statistical distance of $O_{d,k}(\epsilon^k \ell^{-1})$ of some $10d(3k+3)$-design.  If we take $\ell$ independent copies of such random variables, calling them $Y_i$ and let
$$
Y := \frac{\sum_{i=1}^\ell\left(\sqrt{1-\delta^2}\right)^{i-1} Y_i}{\sqrt{\sum_{i=1}^\ell (1-\delta^2)^i}}.
$$
then $Y$ can be generated from seed length
$$
O_{d,k}(\log(n/\epsilon) \ell ) = O_{d,k}(\log(n)\epsilon^{-1}),
$$
and has statistical distance at most $O(\epsilon^k)$ from $Z$.  Thus
$$
\E[f(Y)] = \E[f(Z)]+O(\epsilon^k) = \E[f(X)] + O_{d,k}(\epsilon^k).
$$
\end{proof}

Changing the value of $\epsilon$ appropriately, we have that

\begin{cor}
Let $d$ be a positive integer and $c,\epsilon>0$.  There exists an explicit pseudorandom generator $Y$ with seed length $O_{c,d}(\log(n)\epsilon^{-c})$ so that for any degree-$d$ polynomial threshold function in $n$ variables, and $X$ an $n$-dimensional Gaussian,
$$
\left| \E[f(X)] - \E[f(Y)]\right| \leq \epsilon.
$$
\end{cor}

\section*{Acknowledgements}

This research was done with the support of an NSF postdoctoral fellowship.


\begin{thebibliography}{[99]}

\bibitem{curciutApp} Richard Beigel \emph{The polynomial method in circuit complexity}, Proc. of 8th Annual Structure in Complexity Theory Conference (1993), pp. 82-95.

\bibitem{anticoncentration} A. Carbery, J. Wright \emph{Distributional and $L^q$ norm inequalities for polynomials over convex bodies in $\R^n$}
Mathematical Research Letters, Vol. 8(3), pp. 233–248, 2001.

\bibitem{IndepHalf} I. Diakonikolas, P. Gopalan, R. Jaiswal, R. Servedio, E. Viola, \emph{Bounded Independence Fools Halfspaces} SIAM
Journal on Computing, Vol. 39(8), p. 3441-3462, 2010.

\bibitem{IndepDeg2} Ilias Diakonikolas, Daniel M. Kane, Jelani Nelson, \emph{Bounded Independence Fools Degree-$2$ Threshold Functions},
Foundations of Computer Science (FOCS), 2010.

\bibitem{GPRG} Daniel M. Kane \emph{A Small PRG for Polynomial Threshold Functions of Gaussians} Symposium on the Foundations Of Computer Science (FOCS), 2011.

\bibitem{kIndep} Daniel M. Kane \emph{$k$-Independent Gaussians Fool Polynomial Threshold Functions}, Conference on Computational Complexity (CCC), 2011.

\bibitem{DD} Daniel M. Kane \emph{A Structure Theorem for Poorly Anticoncentrated Gaussian Chaoses and Applications to the Study of Polynomial Threshold Functions}, manuscript http://arxiv.org/abs/1204.0543.

\bibitem{learningApp} Adam R. Klivans, Rocco A. Servedio \emph{Learning DNF in time $2^{O(n^{1/3})}$}, J. Computer and System Sciences Vol. 68 (2004), p. 303-318.

\bibitem{MZ} Raghu Meka, David Zuckerman \emph{Pseudorandom generators for polynomial threshold functions}, Proceedings of the 42nd ACM Symposium on Theory Of Computing (STOC 2010).

\bibitem{hypercontractivity} Nelson \emph{The free {M}arkov field}, J. Func. Anal. Vol. 12(2), p. 211-227, 1973.

\bibitem{commApp} Alexander A. Sherstov \emph{Separating AC0 from depth-2 majority circuits}, SIAM J. Computing Vol. 38 (2009), p. 2113-2129.

\end{thebibliography}
\end{document}